  \documentclass[runningheads,envcountsect]{llncs}      

\usepackage{graphicx,amssymb,color,ifthen}

\newif\ifcompleteproof
\completeprooftrue

  \renewcommand \iff    {\Longleftrightarrow}           

  \newcommand \RelView  {{\sc Rel\-View}}               
  \newcommand \Forall   {{\forall}\,}
  \newcommand \Exists   {{\exists}\,}
\renewcommand \L        {{\sf L}}
\renewcommand \O        {{\sf O}}
  \newcommand \I        {{\sf I}}
  \newcommand \Member   {\mbox{\sf M}}
  \newcommand \M        {\mbox{\sf M}}
  
  \newcommand \Size     {{\sf S}}
  \newcommand \Neg[1]   {{\,\overline{#1}\,}}
  \newcommand \Comp     {;}
  
  \newcommand \Transp[1]{#1^{\sf T}}
  
  \newcommand \Rtup[2]  {{[{#1},{#2}]\!]}}
  \newcommand \Inj      {\mbox{\sl inj}}
  
  \newcommand \ONE      {{\bf 1}\!\!\!{\bf 1}}  
  
  \newcommand \NAT      {{\mathbb N}}

  \newcommand \REL[2]   {\mbox{${#1}$$\,\leftrightarrow\,$${#2}$}}
  \newcommand \VEC[1]   {\mbox{${#1}$$\,\leftrightarrow\,$${\ONE}$}}

  \newcommand \Rel      {\mbox{\sl rel\/}}
  \newcommand \Vec      {\mbox{\sl vec\/}}
  \newcommand \Syq      {\mbox{\sl syq\/}}

  \newcommand \FS       {\mathfrak{S}}
  
  \newcommand \PNAA     {\mbox{$2^N$$\!\times\!$$A^2$}}
  \newcommand \XY       {\mbox{$X$$\times$$Y$}}
  
  \newcommand \AaAa     {\mbox{$A^2$$\times$$A^2$}}
  \newcommand \sol      {\mbox{\textit{sol}}}
  \newcommand \cand     {\mbox{\textit{cand}}}

\newtheorem{Theorem}{Theorem}[section]
\newtheorem{Lemma}{Lemma}[section]
\newtheorem{Definition}{Definition}[section]
\spnewtheorem{Example}{Example}[section]{\bfseries}{\rmfamily}

\newcommand{\complexityclassname}[1]{\ensuremath{\mathrm{#1}}}
\newcommand{\NP}{\complexityclassname{NP}}

\begin{document}

\title{Relation-algebraic and Tool-supported Control of Condorcet Voting}
\titlerunning{Relation-algebraic and Tool-supported Control of Condorcet Voting}
\author{Rudolf Berghammer and Henning Schnoor}
\institute{Institut f\"ur Informatik, Christian-Albrechts-Universit\"at Kiel \\
       Olshausenstra\ss{}e 40, 24098 Kiel, Germany \\
       }

\date{}
\maketitle

\begin{abstract}
We present a relation-algebraic model of Condorcet voting and, based on it,
relation-algebraic solutions of the constructive control problem via 
the removal of voters.
We consider two winning conditions, viz.\ to be a Condorcet 
winner and to be in the (Gilles resp.\ upward) uncovered set.
For the first condition the control problem is known to be NP-hard; for 
the second condition the NP-hardness of the control problem is shown in 
the paper.
All relation-algebraic specifications we will develop in the paper immediately
can be translated into the programming language of the BDD-based computer system
\RelView.
Our approach is very flexible and especially appropriate for prototyping and 
experimentation, and as such very instructive for educational purposes.
It can easily be applied to other voting rules and control problems.
\end{abstract}

\section{Introduction}\label{Sec1} 

Elections have been studied by scientists from different disciplines for more 
than a thousand years. 
In addition to the obvious moral and political issues, elections also give rise 
to several computational questions, which are studied in the field of Computational 
Social Choice. 
The most prominent of these questions is the requirement of an algorithm that 
efficiently computes the winner(s) of an election. 
Surprisingly, such algorithms do not exist for all natural election systems,
see \cite{HemaspaandraHemaspaandraRothe}
for an example. 
However, elections also give rise to computational problems which ideally should 
be \emph{hard} to solve:

\begin{itemize}
\item The \emph{manipulation problem}
      (see \cite{BartholdiOrlin})
      asks to determine a way for a group of voters to vote that serves their interest 
      best, even though the vote might not represent their true preferences. 
      Unfortunately, classical results show that every reasonable voting system gives voters 
      incentives to vote strategically in this way (Gibbard-Satterthwaite theorem,
      cf. \cite{Gibbard,Satterthwaite}). 
\item The \emph{control problem} (see, e.g.,\cite{BarTovTri}) asks for determining 
      a way for the coordinator of an election to set up the election in a way that 
      serves his or her personal interest. 
      In order to achieve this, the coordinator might remove or add alternatives or voters 
      from the election or partition  the election.
\end{itemize}
Following the above-mentioned paper \cite{BartholdiOrlin}, 
numerous papers have studied the complexity of manipulation and control problems for elections 
(see, e.g., 
\cite{ConitzerSandholgLang,FaHeSc,HemHemRot}). 
For many election systems, it can be shown that the studied control or manipulation problem is 
\NP-hard, and thus the election system is deemed to be `secure' against this attempt 
to influence the outcome of the election.
However, it has long been observed that efficient algorithms that work for \emph{many} cases 
can still exist for \NP-hard problems, the very successful history of SAT solvers being 
an impressive example. 
In the context of Computational Social Choice, 
\cite{ConitzerSandholm} 
demonstrates a fast and very simple algorithm that works correctly on `most' inputs 
(according to a suitably chosen probability distribution) and is allowed to compute 
an incorrect result on the remaining inputs.

In this paper, we study an alternative approach to show that \NP-hard election problems may 
be solvable in practice.
We apply the Computer Algebra system \RelView~(see \cite{BerNeu,RELVIEW}), which uses 
mathematical tools from relation algebra in the sense of \cite{SchStr,Sch}, to implement 
algorithms for the control problem of an election. 
Our implementations are provably correct for all instances; hence, as the problems we 
study are \NP-hard, our algorithms do not run in polynomial time in general. 
Instead, we rely on \RelView's optimization to exploit the simple structure of most 
practical instances of the problems we study, which allows for an algorithmic treatment.

Concretely, we study the following problem: 
Given an election consisting of a set of alternatives (also sometimes called candidates), 
voters along with information on how they will vote, and a prefered alternative $a^*$, 
determine a minimum set $Y$ of voters such that removing all voters in $Y$ makes $a^*$ win 
the election. 
The election system we study is the Condorcet voting system with the uncovered set winning 
condition (in case there is no Condorcet winner).
To the best of our knowledge, this is the first paper where a relation-algebraic approach 
is used to solve problems related to elections that directly take the individual votes 
into account. 
An advantage of our approach is that it is very general 
and allows to treat related problems for different election systems with only small 
modifications.
In particular, we could also treat elections in a generalized setting, where voters' 
preferences are not linear orders (such a setting is studied 
in~\cite{FaliszewskiHemaspaandraHemaspaandraRothe}). 
A further advantage is that the correctness proofs for our algorithms are formalized 
in such a way that, in principle, their automatic verification is possible.
Our results and the performance of our algorithms demonstrate that Computer Algebra tools 
can be used successfully to solve \NP-hard problems, where the data structures used 
in the Computer Algebra package automatically allow to exploit the `easyness' that 
may be present in practical instances.
In our case, \RelView\ uses BDDs to efficiently represent relations that are exponential 
in the input size.
Thus, relation-algebraic algorithms can be obtained without specific knowledge 
about the problem domain.


\section{Relation-algebraic Preliminaries}\label{Sec2} 

Given sets $X$ and $Y$, we write $R : \REL{X}{Y}$ if $R$ is a (binary) 
relation with source $X$ and target $Y$, i.e., a subset of  $\XY$.
If the sets of $R$'s \emph{type} $\REL{X}{Y}$ are finite, then
we may consider $R$ as a Boolean matrix.
Since such an interpretation is well suited for ma\-ny purposes and 
also used by \RelView\ as the main possibility to visualize relations, 
in this paper we frequently use matrix terminology and notation.  
Especially, we speak about the entries, rows and columns of a relation/matrix 
and write $R_{x,y}$ instead of $(x,y) \in R$ or $x\,R\,y$.
We assume the reader to be familiar with the ba\-sic operations on relations, 
viz.~$\Transp{R}$ (\emph{transposition}), $\Neg{R}$ (\emph{complement}), 
$R \cup S$ (\emph{union}), $R \cap S$ (\emph{intersection}) and $R \Comp S$ 
(\emph{composition}), the predicates $R \subseteq S$ (\emph{inclusion}) 
and $R = S$ (\emph{equality}), and the special re\-la\-ti\-ons $\O$ 
(\emph{empty relation}), $\L$ (\emph{universal relation}) and $\I$ 
(\emph{identity relation}).
In case of $\O$, $\L$ and $\I$ we overload the symbols, i.e., avoid the binding 
of types to them.

For $R : \REL{X}{Y}$ and $S : \REL{X}{Z}$, by
$\Syq(R,S) = \Neg{\Transp{R}\Neg{S}} \cap \Neg{\Neg{\Transp{R}}S}$
their \emph{symmetric quo\-ti\-ent} $\Syq(R,S) : \REL{Y}{Z}$ is defined.
In the present paper we will only use its point-wise description, saying that 
for all $y \in Y$ and $z \in Z$ it holds $\Syq(R,S)_{y,z}$ iff for all $x \in X$ 
the relationships $R_{x,y}$ and $S_{x,z}$ are equivalent.

In relation algebra \emph{vectors} are a well-known means to model subsets of a given set $X$.
Vectors are relations $r : \VEC{X}$ (we prefer in this context lower case 
letters) with a specific singleton set $\ONE = \{\bot\}$ as target.
They can be considered as Boolean column vectors.
To be consonant with the usual notation, we omit always the 
second subscript, i.e., write $r_x$ instead of $r_{x,\bot}$.  
Then $r$ \emph{describes} the subset $Y$ of $X$ if for all $x \in X$ it
holds $r_x$ iff $x \in Y$.
A point $p : \VEC{X}$ is a vector with precisely one 1-entry.
Consequently, it describes a singleton subset $\{x\}$ of $X$ and we then say 
that it describes the element $x$ of $X$.
If $r : \VEC{X}$ is a vector and $Y$ the subset of $X$ it describes,
then $\Inj(r) : \REL{Y}{X}$ denotes the \emph{embedding relation} of $Y$ 
into $X$. 
In Boolean matrix terminology this means that $\Inj(r)$ is obtained from 
$\I : \REL{X}{X}$ by deleting all rows which do not correspond to an 
element of $Y$ and point-wisely this means that for all $y \in Y$ and $x \in X$
it holds $\Inj(r)_{y,x}$ iff $y = x$.

In conjunction with powersets $2^X$ we will use \emph{membership relations} 
$\Member : \REL{X}{2^X}$ and \emph{size comparison relations} 
$\Size : \REL{2^X}{2^X}$.
Point-wisely they are defined for all $x \in X$ and $Y, Z \in 2^X$ as follows:
$\Member_{x,Y}$ iff $x \in Y$ and $\Size_{Y,Z}$ iff $|Y| \leq |Z|$.
A combination of $\M$ with embedding relations allows a \emph{column-wise 
enumeration} of an arbitrary subset $\FS$ of $2^X$.
Namely, if the vector $r : \VEC{2^X}$ describes $\FS$ in the sense defined 
above and we define $S = \Member \Comp \Transp{\Inj(r)}$, then we get 
$\REL{X}{\FS}$ as type of $S$ and that for all $x \in X$ and $Y \in \FS$ 
it holds $S_{x,Y}$ iff $x \in Y$. 
In the Boolean matrix model this means that the sets of $\FS$ are precisely 
described by the columns of $S$, if the columns are considered as vectors 
of typs $\VEC{X}$.

To model direct products $\XY$ of sets $X$ and $Y$ relation-algebraically, 
the \emph{projection relations} $\pi : \REL{\XY}{X}$ and $\rho : \REL{\XY}{Y}$ 
are the convenient means.
They are the relational variants of the well-known projection functions and,
hence, fulfil for all $u \in \XY$, $x \in X$ and $y \in Y$ the following
equivalences:
$\pi_{u,x}$ iff $u_1 = x$ and $\rho_{u,y}$ iff $u_2 = y$.
Here $u_1$ denotes the first component of $u$ and $u_2$ the second component.
As a general assumption, in the remainder of the paper we always assume a 
pair $u$ to be of the form $u = ( u_1, u_2 )$.
Then $\hat{u}$ denotes the \emph{transposed pair} $(u_2,u_1)$.
The projection relations enable us to specify the well-known pairing operation
of functional programming relation-algebraically.
The \emph{pairing} of $R : \REL{Z}{X}$ and $S : \REL{Z}{Y}$ is defined as 
$\Rtup{R}{S} = R \Comp \Transp{\pi} \cap S \Comp \Transp{\rho} : \REL{Z}{\XY}$.
where $\pi$ and $\rho$ are as above.
Point-wisely this definition says that 
$\Rtup{R}{S}_{z,u}$ iff $R_{z, u_1}$ and $S_{z, u_2}$,
for all $z \in Z$ and $u \in \XY$.
Based on $\pi$ and $\rho$ we are also able to establish a bijective correspondence 
between the relations of type $\REL{X}{Y}$ and the vectors of type $\VEC{\XY}$.
The transformation of $R : \REL{X}{Y}$ into its \emph{corresponding vector}
$\Vec(R) : \VEC{\XY}$ is given by 
$\Vec(R) = (\pi \Comp R \cap \rho) \Comp \L$
and the step back from $r : \VEC{\XY}$ to its \emph{corresponding relation} 
$\Rel(r) : \REL{X}{Y}$ by
$\Rel(r) = \Transp{\pi} \Comp (\rho \cap r \Comp \L)$.
Point-wisely this means that for all $u \in \XY$ the following equivalences
are true:
$\Vec(R)_u$ iff $R_{u_1, u_2}$ and $\Rel(r)_{u_1, u_2}$ iff $r_u $.

\section{A Relation-algebraic Model of Condorcet Voting}\label{Sec3} 

Usually, an election consists of a non-empty and finite set $N$ of voters 
(agents), normally $N = \{1,\ldots,n\}$, a non-empty and finite set $A$ of 
alternatives (candidates), the individual preferences (choices, wishes) of the 
voters and a voting rule that aggregates the winners from the individual 
preferences.
A well-known voting rule is the \emph{Condorcet voting rule}.
Here it is usually assumed that each voter ranks the alternatives from top to 
bottom, i.e., the individual preferences of the voters $i \in N$ are expressed 
via linear strict orders $>_i$ $:$ $\REL{A}{A}$.
{From} them the \emph{dominance relation} $C : \REL{A}{A}$ is computed that 
specifies the collective preferences.
An \emph{instance} of a Condorcet election consists of the sets $N$, $A$, 
and the relations $>_i$ for all $i\in N$.
In the following we consider the approach that $C_{a,b}$ iff the number 
of voters $i$ with $a >_i b$ is (strictly) greater than the number 
of voters $i$ with $b >_i a$.  
In this case we also say that $a$ \emph{beats} $b$ with $p$ points, where $p$ is 
the (positive) difference between these numbers.
It is known that $C$ may contain cycles and that an alternative that dominates 
all other ones -- a so-called \emph{Condorcet winner} -- does not necessarily exist.
To get around this problem, in the literature so-called \emph{choice sets}
have been introduced which take over the role of the best alternative and 
specify the winners (see e.g., \cite{Las} for more details).
In this paper, we will study the choice set \emph{Uncovered Set}.

For a relation-algebraic treatment of Condorcet voting, we first model its 
input, i.e., the individual preferences of the voters, accordingly.

\begin{Definition}\label{DefPC}
The relation $P : \REL{N}{A^2}$ \emph{models} the instance 
$(N,A,(<_i)_{i\in N})$ of a Condorcet election if $P_{i,u}$ is equivalent 
to $u_1 >_i u_2$, for all $i \in N$ and $u \in A^2$.
\end{Definition}

\noindent
In the following \RelView\ picture an input relation $P$ is shown.
The labels of the rows and columns indicate that the voters are the natural
numbers from 1 to 13 and the alternatives are the eight letters from $a$ 
to $h$.
$$
\includegraphics[scale=1.109]{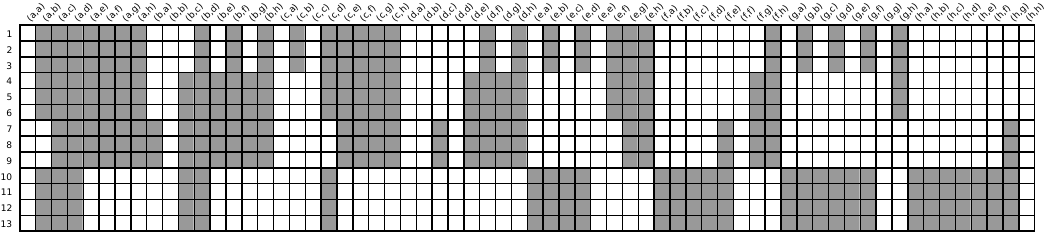}
$$
It is troublesome to identify from this picture the individual preferences.
But if we select the single rows, transpose them to obtain vectors of 
type $\VEC{A^2}$ and apply the function $\Rel$ of Section \ref{Sec2} to 
the latter, then \RelView\ depicts the individual preferences as Boolean 
matrices.
For the rows 1, 4, 7 and 10 we get, in the same order, the following 
Boolean matrices for $>_1, >_4, >_7$ and $>_{11}$:
$$
\includegraphics[scale=0.249]{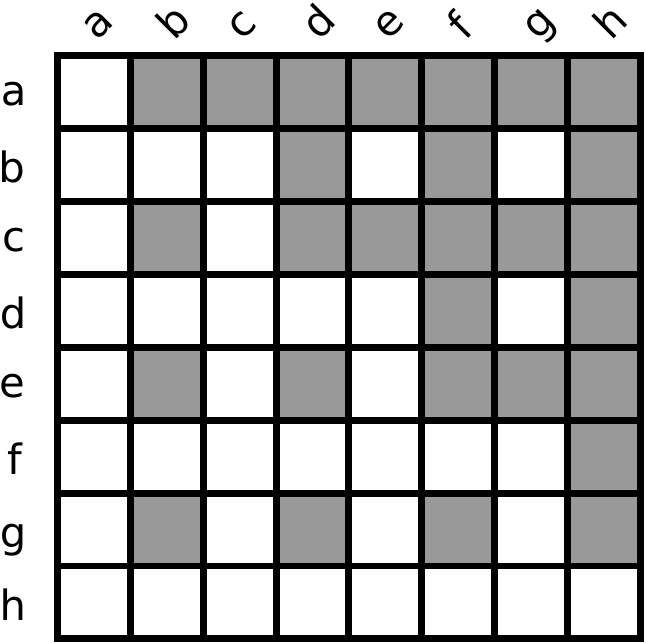} \qquad
\includegraphics[scale=0.249]{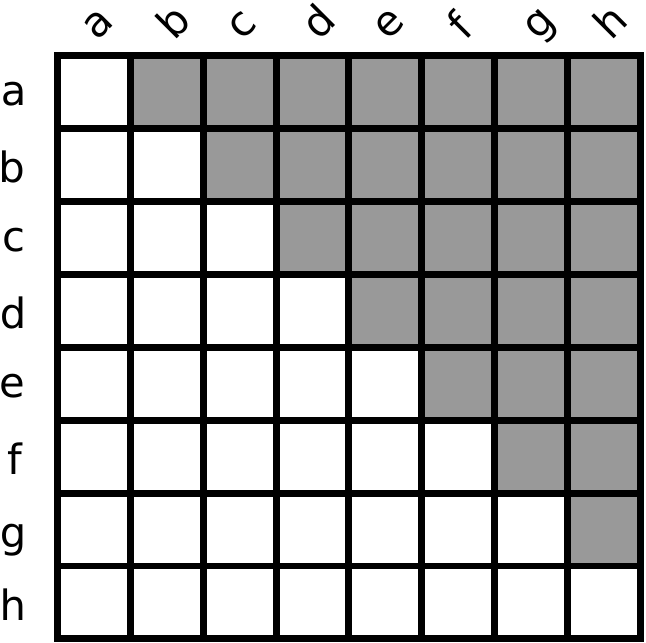} \qquad
\includegraphics[scale=0.249]{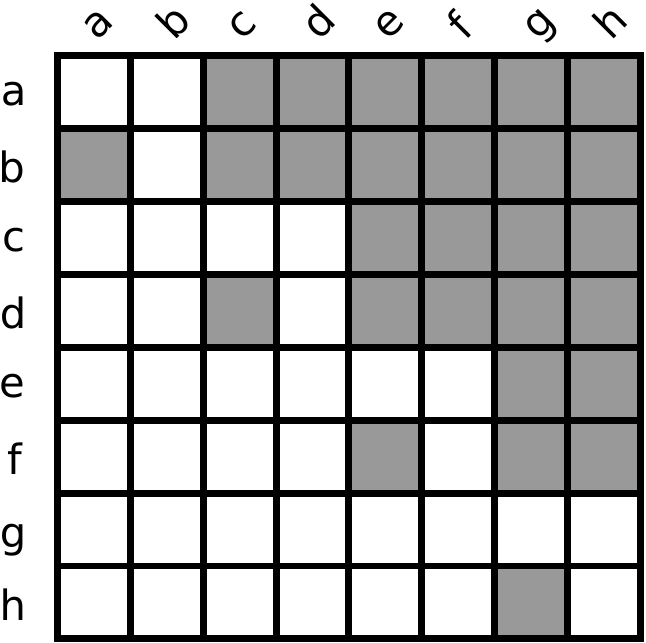} \qquad
\includegraphics[scale=0.249]{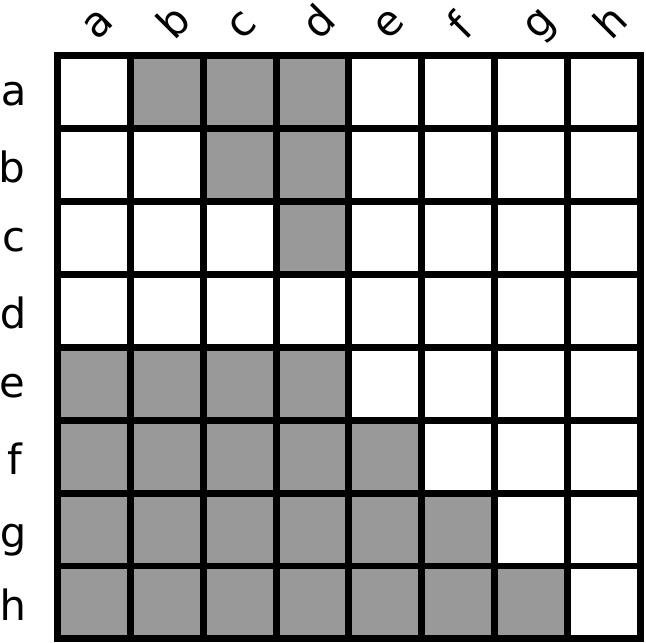}
$$
Now, the preferences of the single voters are easy to see\footnote{%
  A still more appropriate method is to compute for each relation $>_i$ its 
  Hasse diagram in the sense of \cite{SchStr} and to draw the latter in 
  \RelView\ as directed graphs.}. 
Voters 1 to 3 rank their alternatives from top to bottom as $a,c,e,g,b,d,f,h$,
voters 4 to 6 as $a,b,c,d,e,f,g,h$, voters 7 to 9 as $b,a,d,c,f,e,h,g$ and
the remaining voters 10 to 13 as $h,g,f,e,a,b,c,d$.
The procedure also shows how to construct, in general, the input 
$P : \REL{N}{A^2}$ from strict orders $>_i$ $;$ $\REL{A}{A}$ by inverting it.
We have to number the voters from 1 to $n$, then to transform each 
relation $>_i$ into $\Transp{\Vec(>_i)} : \REL{\ONE}{A^2}$, i.e., the 
transpose of its corresponding vector, and finally to combine the transposed 
vectors row by row into a Boolean matrix.
The latter means that we have to form the relation-algebraic sum
$\Transp{\Vec(>_1)} + \cdots + \Transp{\Vec(>_n)}$.
We won't to go into details with regard to sums of relations and refer to 
\cite{Sch}, where a relation-algebraic specification 
via injection relations is given.
Instead, we demonstrate how to get from the individual preferences relation $P$ 
the collective preferences, i.e., the dominance relation $C$.
In what follows, we assume the projection relations $\pi, \rho : \REL{A^2}{A}$ 
of the direct product $A^2$ to be at hand as well as the membership relation 
$\M : \REL{N}{2^N}$ and the size comparison relation $\Size : \REL{2^N}{2^N}$.
Each of these relations is available in \RelView\ via a pre-defined function
and their BDD-implementations are rather small.
See \cite{Leoniuk,Milanese} for details.

\begin{Theorem}\label{CV1}
Suppose that $P : \REL{N}{A^2}$ models an instance of Condorcet voting.
If we specify relations $E, F : \REL{A^2}{2^N}$ and $C : \REL{A}{A}$ by
$$
E = \Syq(P,\M)
\quad
F = \Syq(P \Comp \Rtup{\rho}{\pi},\M)
\quad
C = \Rel((E \cap F \Comp (\Size \cap \Neg{\Transp{\Size}})) \Comp \L),
$$
then $C_{u_1,u_2}$ is equivalent to
$|\{i \in N \mid P_{i,u}\}| > |\{i \in N \mid P_{i,\hat{u}}\}|$,
for all $u \in A^2$.
\end{Theorem}
\begin{proof}
For the given $u \in A^2$ we prove in a preparatory step for all $Y \in 2^N$ 
that
$$
\begin{array}{rl}
       E_{u,Y} 
\iff & \Syq(P,\M)_{u,Y} \\
\iff & \Forall i \in N : 
       P_{i,u} \leftrightarrow \M_{i,Y} \\[0.99mm]
\iff & \Forall i \in N : 
       P_{i,u} \leftrightarrow i \in Y \\[0.99mm]
\iff & \{i \in N \mid P_{i,u}\} = Y.
\end{array}
$$
Using that the exchange relation $\Rtup{\rho}{\pi} : \REL{A^2}{A^2}$ relates 
the pair $u$ precisely with its transposition $\hat{u} = (u_2,u_1)$, in a rather 
similar way we can prove that for all $Z \in 2^N$ the following property holds:
$$
F_{u,Z} \iff \{i \in N \mid P_{i,\hat{u}}\} = Z
$$
By means of these two auxiliary results, we now conclude
the proof as follows:
$$
\begin{array}[b]{rl@{}l}
~      C_{u_l,u_2}
\iff & \Rel((E \cap F \Comp (\Size \cap \Neg{\Transp{\Size}})) \Comp \L
           )_{u_1,u_2} & \\[0.99mm]
\iff & ((E \cap F \Comp (\Size \cap \Neg{\Transp{\Size}})) \Comp \L
       )_{u} \\[0.99mm]
\iff & \Exists Y \in 2^N : 
               E_{u,Y} \wedge 
               (F \Comp (\Size \cap \Neg{\Transp{\Size}}))_{u,Y} \wedge 
               \L_Y & \\[0.99mm]
\iff & \Exists Y \in 2^N : 
               E_{u,Y} 
               \wedge 
               \Exists Z \in 2^N : 
                   F_{u,Z} \wedge \Size_{Z,Y} \wedge \neg \Size_{Y,Z} & \\[0.99mm]
\iff & \Exists Y \in 2^N : 
               E_{u,Y} 
               \wedge 
               \Exists Z \in 2^N : 
                   F_{u,Z} \wedge |Z| \leq |Y| \wedge |Y| >|Z| & \\[0.99mm]
\iff & \Exists Y, Z \in 2^N :                        
               \{i \in N \mid P_{i,u}\} = Y 
               \wedge
               \{i \in N \mid P_{i,\hat{u}}\} = Z 
               \wedge 
               |Z| < |Y| & \\[0.99mm]
\iff & |\{i \in N \mid P_{i, u}\}| > |\{i \in N \mid P_{i,\hat{u}}\}| & ~\Box
\end{array}
$$
\end{proof}

\noindent
The specifications of Theorem \ref{CV1} can be executed by means of \RelView\ 
after a straightforward translation into its programming language.
In case of the above input relation $P$ the tool computed the
following dominance relation $C$.
{From} the first row of $C$ we see that alternative $a$ is tha Condorcet winner
since it dominates all other alternatives.

\medskip

\centerline{
\includegraphics[scale=0.249]{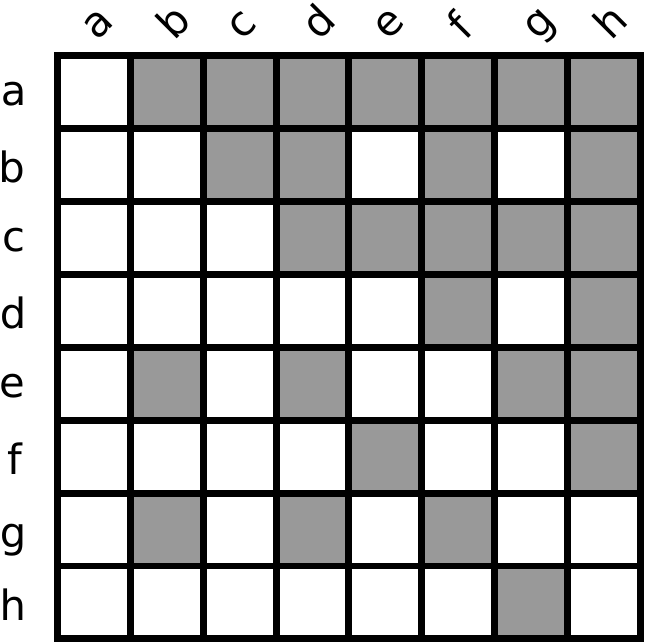}
}

\medskip

\noindent
This relation is not only asymmetric (i.e., satisfies $C \cap \Transp{C} = \O$) 
but also \emph{complete} (i.e., satisfies $\Neg{\I} \subseteq C \cup \Transp{C}$).
Altogether, $C$ is a \emph{tournament} relation and this property implies the 
uniqueness of a Condorcet winner in the case that one exists.
How to compute, in general, from the dominance relation $C$ the choice sets using
relation-algebraic means is demonstreted in \cite{BerRusSwa}.

\section{Control of Condorcet Voting by Deleting Voters}\label{Sec4} 

We only consider the constructive variant of the control problem for Condorcet 
voting, where control is done by deleting voters.
Usually, the task is formulated as a minimization-problem:
Given a specific alternative $a^*$, determine a minimum set of voters 
$Y$ such that the removal of $Y$ from the set $N$ of all voters makes $a^*$ 
to a winner\footnote{%
 The destructive veriant of our control problem asks for a minimum set of voters 
 the removal of which prevents win of $a^*$.
 }.
To allow for an easier relation-algebraic representation, we consider the dual
maximization-problem, i.e., we ask for a maximum
set of voters $X$ such that $a^*$ wins subject to the condition that only 
voters from $X$ are allowed to vote.
It is obvious that from $X$ then a desired $Y$ is obtained via $Y = N \setminus X$.

We start with the assumption that `to win' means `to be a Condorcet winner'.
As shown in \cite{BarTovTri}, Condorcet voting is computationally resistant to our 
control type in
case of this specification of winners.
I.e., it is NP-hard to decide, for $a^* \in A$ and $k \in \NAT$ as inputs, whether 
it is possible to find $k$ voters whose removal makes $a^*$ to a Condorcet winner.

As a first step towards a solution of the maximization-problem, we relativize the 
dominance relation $C$ by additionally considering the sets of 
voters $X$ which only are allowed to vote.
Concretely this means that we specify a relation $R$ that relates $X \in 2^N$ 
with $a, b \in A$ iff $|\{i \in X \mid a >_i b\}| > |\{i \in X \mid b >_i a\}|$.
Since we work with binary relations, we have to combine two of the three objects 
$X$, $a$ and $b$ to a pair.
We do this with $a$ and $b$, i.e., relate $X$ with $u$ under the assumption 
that $u_1$ equals $a$ and $u_2$ equals $b$.
Then the following theorem shows how
the \emph{relativized dominance relation} $R : \REL{2^N}{A^2}$ can 
be specified relation-algebraically.
Again we assume the relations $\pi, \rho : \REL{A^2}{A}$, $\M : \REL{N}{2^N}$ and 
$\Size : \REL{2^N}{2^N}$ to be at hand.


\begin{Theorem}\label{CV2}
Suppose again that $P : \REL{N}{A^2}$ models an instance of Condorcet voting.
If we specify relations $E, F : \REL{\PNAA}{2^N}$ and $R : \REL{2^N}{A^2}$ 
by
$$
E = \Syq( \Rtup{\M}{P} , \M)
\quad
F = \Syq(\Rtup{\M}{P \Comp \Rtup{\rho}{\pi}},\M)
\quad
R = \Rel((E \cap F \Comp (\Size \cap \Neg{\Transp{\Size}})) \Comp \L),
$$
then $R_{X,u}$ is equivalent to
$|\{i \in X \mid P_{i,u}\}| > |\{i \in X \mid P_{i,\hat{u}}\}|$,
for all $X \in 2^N$ and $u \in A^2$.
\end{Theorem}

\begin{proof}
Assume arbitrary objects $X \in 2^N$ and $u \in A^2$ to be given.
Then, we have for all $Y \in 2^N$ the following equivalence:
$$
\begin{array}{rl}
       E_{(X,u),Y} 
\iff & \Syq(\Rtup{\M}{P},\M)_{(X,u),Y} \\
\iff & \Forall i \in N : \Rtup{\M}{P}_{i,(X,u)} 
                         \leftrightarrow
                         \M_{i,Y} \\[0.99mm]
\iff & \Forall i \in N : \M_{i,X} \wedge P_{i,u}
                         \leftrightarrow
                         \M_{i,Y} \\[0.99mm]
\iff & \Forall i \in N : i \in X \wedge P_{i,u}
                         \leftrightarrow
                         i \in Y \\[0.99mm]
\iff & \{i \in X \mid P_{i,u}\} = Y
\end{array}
$$
In a similae way we can show for all $Z \in 2^N$ the following fact,
using the property of the exchange relation $\Rtup{\rho}{\pi} : \REL{A^2}{A^2}$ 
mentioned in the proof of Theorem \ref{CV1}:
$$
F_{(X,u),Z} \iff \{i \in X \mid P_{i,\hat{u}}\} = Z.
$$
Now, the following calculation shows the claim:
$$
\begin{array}[b]{rl@{}l}
~~     R_{X,u}
\iff & \Rel((E \cap F \Comp (\Size \cap \Neg{\Transp{\Size}}) ) \Comp \L
           )_{X,u} & \\[0.99mm]
\iff & ((E \cap F \Comp (\Size \cap \Neg{\Transp{\Size}}) ) \Comp \L
       )_{(X,u)} \\[0.99mm]
\iff & \Exists Y \in 2^N : 
               E_{(X,u),Y} \wedge 
               (F \Comp (\Size \cap \Neg{\Transp{\Size}}) )_{(X,u),Y} 
               \wedge \L_Y & \\[0.99mm]
\iff & \Exists Y \in 2^N : 
               E_{(X,u),Y} 
               \wedge 
               \Exists Z \in 2^N : 
                  F_{(X,u),Z} \wedge \Size_{Z,Y} \wedge \neg \Size_{Y,Z} & \\[0.99mm]
\iff & \Exists Y, Z \in 2^N : 
              \{i \in X \mid P_{i,u}\} = Y 
              \wedge
              \{i \in X \mid P_{i,\hat{u}}\} = Z 
              \wedge 
              |Z| < |Y| & \\[0.99mm]
\iff & |\{i \in X \mid P_{i, u}\}| > |\{i \in X \mid P_{i,\hat{u}}\}| & ~\Box
\end{array}
$$
\end{proof}

\noindent
In the second step, we now take the relativized dominance relation $R$ of 
Theorem \ref{CV2} and specify with its help a vector $\cand : \VEC{2^N}$ 
that describes the subset of $2^N$ the members of which are the sets $X$ which 
are \emph{candidates} for the solution of our control problem.
The latter property means that $a^*$ is a Condorcet winner, provided that 
only voters from $X$ are allowed to vote.
{From} the vector $\cand$ we then finally compute the vector description 
$\sol : \VEC{2^N}$ of the maximum candidate sets, which are the solutions we 
are looking for.
The next theorem shows how to get $\cand$ and $\sol$ from $R$ and $a^*$.

\begin{Theorem}\label{CV3}
Suppose that $R : \REL{2^N}{A^2}$ is the relation specified in Theorem \ref{CV2} 
and that the specific alternative $a^* \in A$ is described by the point 
$p : \VEC{A}$.
If we specify vectors $\cand, \sol : \VEC{2^N}$ by
$$
\cand = \Neg{\Neg{R} \Comp (\pi \Comp p \cap \Neg{\rho \Comp p}) }
\qquad
\sol = \cand \cap \Neg{ \Neg{\Transp{\Size}} \Comp \cand} ,
$$
then the set 
$\{X \in 2^N \mid \Forall b \in A \setminus \{a^*\} : |\{i \in X \mid P_{i, (a^*,b)}\}| 
> |\{i \in X \mid P_{i, (b,a^*)}\}|\}$
is described by $\cand$ and the set of its maximum sets by $\sol$.
\end{Theorem}

\begin{proof}
Since $p$ describes $a^*$, for all $u \in A^2$ we have 
$(\pi \Comp p)_u$ iff $u_1 = a^*$ and $\Neg{\rho \Comp p}_u$ iff $u_2 \not= a^*$.
We now assume an arbitrary set $X \in 2^N$ and calculate as follows, where in 
the fifth step Theorem \ref{CV2} is applied:
$$
\begin{array}{rl}
       \cand_X
\iff & \Neg{\Neg{R} \Comp ( \pi \Comp p \cap \Neg{\rho \Comp p} )}_X \\[0.99mm]
\iff & \neg \Exists u \in A^2 : 
            \Neg{R}_{X,u} \wedge 
            (\pi \Comp p)_u \wedge \Neg{\rho \Comp p}_u \\[0.99mm]
\iff & \neg \Exists u \in A^2 : 
            \Neg{R}_{X,u} \wedge u_1 = a^* \wedge u_2 \not= a^* \\[0.99mm]
\iff & \Forall u \in A^2 : 
            u_1 = a^* \wedge u_2 \not= a^* \rightarrow R_{X,u} \\[0.99mm]
\iff & \Forall u \in A^2 : 
            u_1 = a^* \wedge u_2 \not= a^* 
            \rightarrow 
            |\{i \in X \mid P_{i,u}\}| > |\{i \in X \mid P_{i,\hat{u}}\}| \\[0.99mm]
\iff & \Forall b \in A : 
            b \not= a^* 
            \rightarrow 
            |\{i \in X \mid P_{i,(a^*,b)}\}| > |\{i \in X \mid P_{i,(b,a^*)}\}| \\[0.99mm]
\iff & \Forall b \in A \setminus \{a^*\} :
            |\{i \in X \mid P_{i,(a^*,b)}\}| > |\{i \in X \mid P_{i,(b,a^*)}\}| \\[0.99mm]
\end{array}
$$
Hence, the first claim follows from the definition of the set a vector describes.
To prove the second claim, we take again an arbitrary set $X \in 2^N$.
Then, we get:
$$
\begin{array}{rl}
       \sol_X
\iff & (\cand \cap \Neg{ \Neg{\Transp{\Size}} \Comp \cand })_X \\
\iff & \cand_X \wedge 
       \Neg{ \Neg{\Transp{\Size}} \Comp \cand }_X \\
\iff & \cand_X \wedge 
       \neg \Exists Y \in 2^N : \Neg{\Size}_{Y,X} \wedge \cand_Y \\
\iff & \cand_X \wedge 
       \Forall Y \in 2^N : \cand_Y \rightarrow \Size_{Y,X} \\
\iff & \cand_X \wedge 
       \Forall Y \in 2^N : \cand_Y \rightarrow |Y| \leq |X|
\end{array}
$$
This equivalence implies that $\sol$ describes the set of maximum sets of voters 
$X$ for which $\cand_X$ holds, that is, for which $a^*$ wins subject to the
condition that only voters from $X$ are allowed to vote.
\hfill $\Box$
\end{proof}
               
\noindent
Using \RelView\ we have solved our control problem with Condorcet winners as
winning alternatives for the above input relation $P$ and each of the eight alternatives.  
The tool showed that only the alternatives $a$, $b$ and $h$ can made to 
Condorcet winners by deleting voters.
Some of the results for these alternatives are presented in the following six 
\RelView\ pictures:
\label{example:relview results}

\centerline{
\includegraphics[scale=0.239]{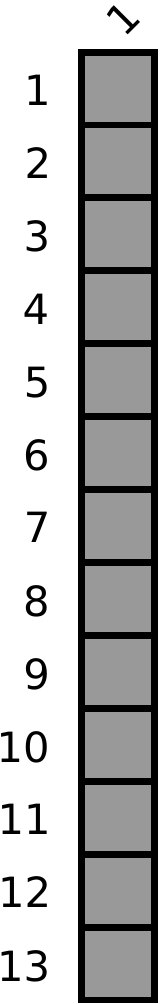} \qquad                     
\includegraphics[scale=0.239]{CCon}    \qquad\quad
\includegraphics[scale=0.239]{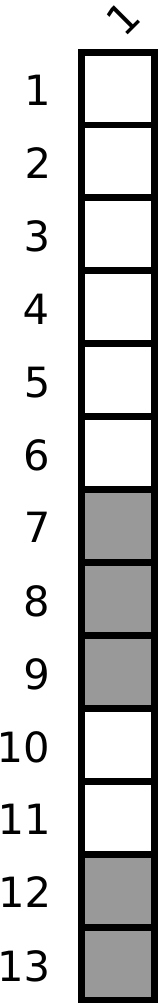} \qquad
\includegraphics[scale=0.239]{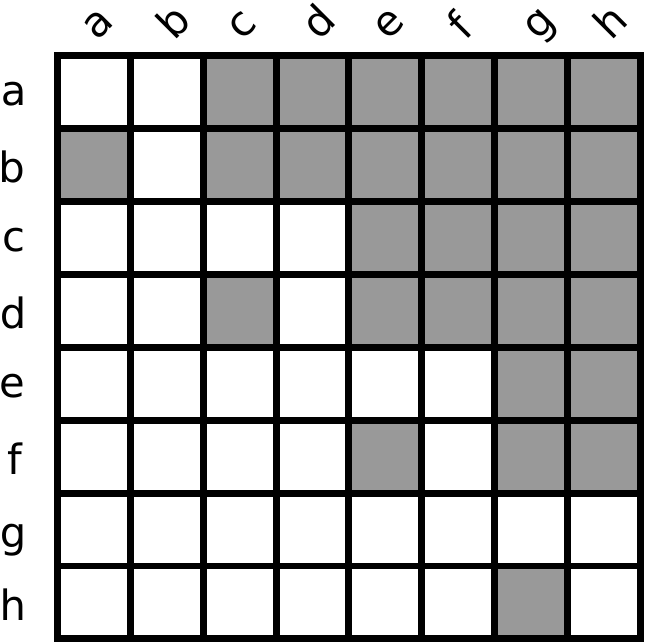}    \qquad\quad
\includegraphics[scale=0.239]{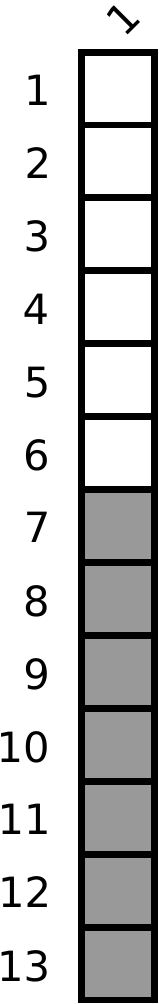} \qquad
\includegraphics[scale=0.239]{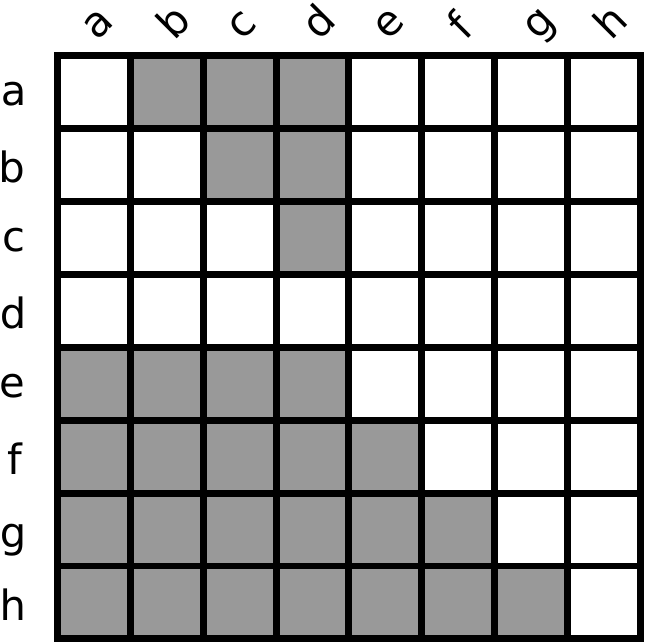} 
}

\medskip

\noindent
The vector on position 1 says that $a$ is a Condorcet winner if all voters 
are allowed to vote and the corresponding dominance relation on position 2 
is the original dominance relation $C$.
To make $b$ to a Condorcet winner at least eight voters must be deleted.
Altogether there are 45 possibilities for this.
The vector on position 3 shows one of them, where the voters from 1 to 6 and 
the voters 10 and 11 are deleted.
On position 4 the resulting dominance relation is depicted.
To get $h$ as Condorcet winner requires a removal of at least six voters.
According to \RelView\ there are 85 possibilities for this.
One of them and the resulting dominance relation are depicted at positions 5 
and 6. 

Since Condorcet winners do not always exist, choice sets have been introduced as 
a general concept that always allows to define the winners of Condorcet 
voting.
In the remainder of this section we treat a well-known example, the 
\emph{uncovered set}.
This choice set is usually defined via an induced transitive subrelation of 
the dominance relation $C$, called \emph{covering relation}.
In the literature different such relations are discussed.
We concentrate on a relation $G : \REL{A}{A}$ that in \cite{Duggan} is called
\emph{Gilles covering} and in \cite{Bra} \emph{upward covering}.
Its usual point-wise definition says that $G_{a,b}$ iff $C_{a,b}$ and for all 
$c \in A$ from $C_{c,a}$ it follows $C_{c,b}$, for all $a,b \in A$.
This relation-algebraically can be specified as equation
$G =  C \cap \Neg{\Transp{C} \Comp \Neg{C}}$.
The (Gilles or upward) \emph{uncovered set} is the set of all $a \in A$ such that
there exists no $b \in A \setminus \{a\}$ with $G_{b,a}$.
It is non-empty because $G$ is a strict-order and $A$ is finite.
To the best of our knowledge, the computational
complexity of control problems for Condorcet elections with winning conditions
different from being a Condorcet winner has not been studied in the literature.
We obtain the first result in this direction by proving that the problem to control 
Condorcet elections with upward covering
by deleting voters is NP-hard (see Section~\ref{Sec5}).
To solve our control problem for this specification of winners we use 
the same idea as in the relativization of the relation $C$ to the relation 
$R$ by additionally considering the set of voters $X$ which are allowed 
to vote.
The next theorem shows how to obtain the \emph{relativized covering relation} $U$
from the relativized dominance relation $R$.


\begin{Theorem}\label{CV4}
Suppose again that $R : \REL{2^N}{A^2}$ is the relation specified in Theorem 
\ref{CV2}.
If we specify relations $E : \REL{\AaAa}{A^2}$ and 
$U : \REL{\PNAA}{2^N}$ by
$$
E = \Transp{\Rtup{\pi \Comp \Transp{\rho}}{\rho \Comp \Transp{\rho}}}
     \cap
     \Vec(\pi \Comp \Transp{\pi}) \Comp \L
\qquad
U = R \cap \Neg{\Rtup{R}{\Neg{R}} \Comp E},
$$
then for all $X \in 2^N$ and $u \in A^2$ we have
$$
U_{X,u} \iff R_{X,u} 
             \wedge 
             \Forall c \in A : R_{X,(c,u_1)} \rightarrow R_{X,(c,u_2)} .
$$
\end{Theorem}
\begin{proof}
Let $X \in 2^N$ and $u \in A^2$ be given.
In a first step we show for all $v, w \in A^2$ the following property, 
where we use the equivalence of $(\pi \Comp \Transp{\rho})_{u,v}$ and 
$u_1 = v_2$, of $(\rho \Comp \Transp{\rho})_{u,w}$ and $u_2 = w_2$, and 
of $(\pi \Comp \Transp{\pi})_{v,w}$ and $v_1 = w_1$:
$$
\begin{array}[b]{rl}
       E_{(v,w).u}
\iff & (\Transp{\Rtup{\pi \Comp \Transp{\rho}}{\rho \Comp \Transp{\rho}}}
        \cap
        \Vec(\pi \Comp \Transp{\pi}) \Comp \L
       )_{(v,w),u} \\[0.99mm]
\iff & \Rtup{\pi \Comp \Transp{\rho}}{\rho \Comp \Transp{\rho}}_{u,(v,w)}
       \wedge
       (\Vec(\pi \Comp \Transp{\pi}) \Comp \L)_{(v,w),u} \\[0.99mm]
\iff & (\pi \Comp \Transp{\rho})_{u,v}
       \wedge
       (\rho \Comp \Transp{\rho})_{u,w}
       \wedge
       \Vec(\pi \Comp \Transp{\pi})_{(v,w)} \\[0.99mm]
\iff & u_1 = v_2
       \wedge
       u_2 = w_2
       \wedge
       (\pi \Comp \Transp{\pi})_{v,w} \\[0.99mm]
\iff & u_1 = v_2
       \wedge
       u_2 = w_2
       \wedge
       v_1 = w_1
\end{array}
$$
We now can calculate as follows to conclude the proof:
$$
\begin{array}[b]{rl@{}l}
~      U_{X,u}
\iff & (R \cap \Neg{\Rtup{R}{\Neg{R}} \Comp E})_{X,u} \\[0.99mm] 
\iff & R_{X,u}
       \wedge
       \Neg{\Rtup{R}{\Neg{R}} \Comp E}_{X,u} & \\[0.99mm]
\iff & R_{X,u} 
       \wedge  
       \neg \Exists v, w \in A^2 : 
            \Rtup{R}{\Neg{R}}_{X,(v,w)} \wedge E_{(v,w),u} & \\[0.99mm]
\iff & R_{X,u} 
       \wedge 
       \neg \Exists v, w \in A^2 :
               R_{X,v} 
               \wedge 
               \Neg{R}_{X,w} 
               \wedge
               u_1 = v_2
               \wedge
               u_2 = w_2
               \wedge
               v_1 = w_1 & \\[0.99mm]
\iff & R_{X,u} 
       \wedge
       \neg \Exists c \in A :
               R_{X,(c,u_1)} 
               \wedge 
               \Neg{R}_{X,(c,u_2)} & \\[0.99mm]
\iff & R_{X,u} 
       \wedge
       \Forall c \in A :
               R_{X,(c,u_1)} 
               \rightarrow
               R_{X,(c,u_2)} & \Box
\end{array}
$$
\end{proof}

\noindent
After this result we are able to solve our control problem also
for the uncovered set as set of winners.
We use again a vector $\cand$ for the description of the candidate sets and a vector 
$\sol$ for the description of the sulutions.

\begin{Theorem}\label{CV5}
Suppose that $U : \REL{2^N}{A^2}$ is the relation specified in Theorem \ref{CV4} 
and that the specific alternative $a^* \in A$ is described by the point 
$p : \VEC{A}$.
If we specify vectors $\cand, \sol : \VEC{2^N}$ by
$$
\cand = \Neg{U \Comp (\Neg{\pi \Comp p} \cap \rho \Comp p)}
\qquad
\sol = \cand \cap \Neg{ \Neg{\Transp{\Size}} \Comp \cand} ,
$$
then the set 
$\{X \in 2^N \mid \neg \Exists b \in A \setminus \{a^*\} : U_{X(b,a^*)}\}$
is described by $\cand$ and the set of its maximum sets by $\sol$.
\end{Theorem}
\begin{proof}
Because $a^*$ is described by $p$, for all $u \in A^2$ we have 
$\Neg{\pi \Comp p}_u$ iff $u_1 \not= a^*$ and $(\rho \Comp p)_u$ iff $u_2 = a^*$.
Now, for all $X \in 2^N$ we can calculate as follows
to show the first claim (for the second claim cf. the proof of Theorem \ref{CV3}).
$$
\begin{array}[b]{rl}
       \cand_X
\iff & \Neg{U \Comp (\Neg{\pi \Comp p} \cap \rho \Comp p)}_X \\[0.99mm]
\iff & \neg \Exists u \in A^2 : 
            U_{X,u} \wedge \Neg{\pi \Comp p}_u \wedge (\rho \Comp p)_u \\[0.99mm]
\iff & \neg \Exists u \in A^2 : 
            U_{X,u} \wedge u_1 \not= a^* \wedge u_2 = a^* \\[0.99mm]
\iff & \neg \Exists b \in A : 
            U_{X,(b,a^*)} \wedge b \not= a^* \\[0.99mm]
\iff & \neg \Exists b \in A \setminus \{a^*\} : U_{X,(b,a^*)} 
\end{array}
\eqno{\Box}
$$
\end{proof}

\noindent
As already mentioned, the uncovered set is always non-empty.
The degenerate case is that no voter is allowed to vote.
Then the resulting dominance relation as well as the induced covering 
relation are empty and, thus, the uncovered set equals $N$.
\RelView\ showed that in our running example this situation occurs 
if $c$ or $d$ shall win.
We already know that $a$ wins without a removal of voters.
By reason of the tool at least five voters must be deleted to ensure 
win for $e, f, g$ or $h$ and the corresponding numbers of possibilities are
11, 111, 15 and 126.
And, finally, $b$ becomes winning if at least seven voters are not allowed
to vote.
To reach the goal there exist 120 possibilities.
We end this section with the following three \RelView\ pictures that concern
alternative $e$:

\medskip

\centerline{
\includegraphics[scale=0.239]{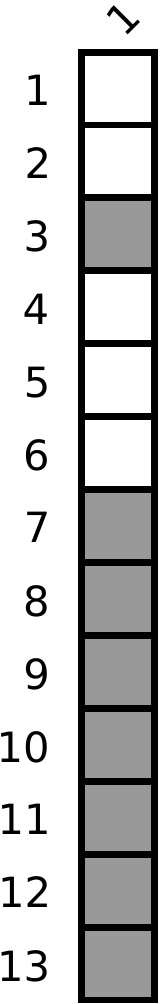}   \qquad           
\includegraphics[scale=0.239]{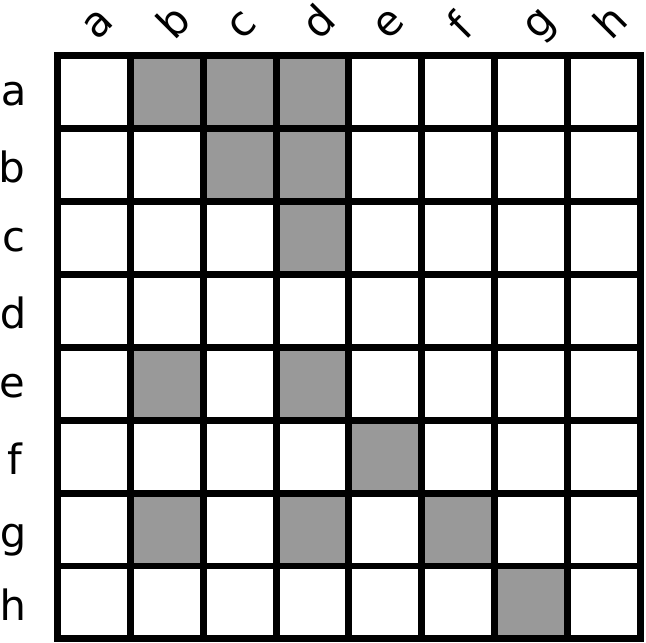}    \qquad
\includegraphics[scale=0.239]{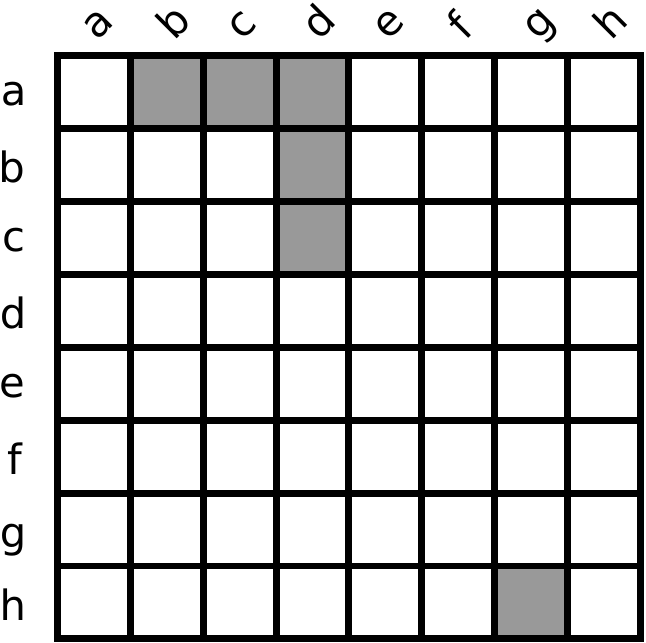}
}

\medskip

\noindent
The vector 
shows that $e$ is in the uncovered set if the voters 1, 2, 4, 5 and 6 are 
deleted, the relation in the middle is the dominance relation 
resulting from this, and the relation on the right is the induced covering 
relation.
The empty columns show that, besides $e$, the removal 
also make $a, f$ and $h$ uncovered.

\section{Control Remains Hard if Uncovered Alternatives Win}\label{Sec5} 

\newcommand{\pointsItem}[1]{%
\ifcompleteproof
\item[#1]
\else
\item{#1}
\fi}

As already mentioned, in \cite{BarTovTri} it is shown that for Condorcet voting 
constructive control by deleting voters is NP-hard if Condorcet winners are 
defined as winners.
In this section we prove that this result remains true if instead of Condorcet 
winners the uncovered alternatives are taken.
To this end we first introduce the following problem that we will be used in 
our reduction.

\newcommand{\mathtext}[1]{\ensuremath{\mathrm{\text{#1}}}}
\newcommand{\prblemname}[1]{\ensuremath{\mathsf{#1}}}
\newcommand{\card}[1]{\left| #1 \right|}
\newcommand{\set}[1]{\ensuremath\left\{#1\right\}}
\newcommand{\oit}{\mbox{1-in-3}}
\newcommand{\oitsat}
  {\ensuremath{\prblemname{1}\textrm{-}\prblemname{in}\textrm{-}\prblemname{3}\textrm{-}\prblemname{Sat}}}
\newcommand{\oitsatversion}{\ensuremath{\oitsat'}}

\begin{Definition}
The problem \textsf{X4C} (\emph{exact cover by $4$-sets}) is the following:
\begin{center}
\begin{tabular}{lp{10.5cm}}
Input: & Sets $S_1,\dots,S_k \in 2^{\{1,\ldots,n\}}$ such
         that for all $i \in \{1,\ldots,k\}$ it holds
         $\card{S_i} = 4$ and 
         $\card{\set{i\ \vert\ j \in S_i}} = 3$ for all $j \in \{1,\ldots,n\}$. \\
Question: & Is there some set $I \in 2^{\set{1,\dots,k}}$ such that 
         $\bigcup_{i\in I} S_i = \{1,\ldots,n\}$ and $S_i \cap S_j=\emptyset$ for 
         all $i,j\in I$ with $i \neq j$?
\end{tabular}
\end{center}
\end{Definition}

\noindent
Note that if an $I$ as required exists, then $\card I = \frac14n$, 
since each $S_i$ has cardinality $4$ and the union must have cardinality 
$n$. 
On the other hand, if an $I$ with $\bigcup_{i\in I}S_i = \{1,\ldots,n\}$ exists and
$\card I = \frac14n$, then by a simple counting argument, 
$S_i \cap S_j = \emptyset$ for all $i,j \in I$ with $i\neq j$. 
Also, the value $k$ in the problem instance must necessarily be equal 
to $\frac34n$, since each $S_i$ has $4$ elements and each $j\in \{1,\ldots,n\}$ 
appears in exactly $3$ of the sets $S_i$. 
In particular, it follows that $n$ is a multiple of $4$ in every 
instance fo \textsf{X4C}.
The following result is mentioned without proof in \cite{FaHeSc}, we give the complete proof:

\begin{Lemma}\label{lemma:x4c}
The problem \textsf{X4C} is NP-hard.
\end{Lemma}

\begin{proof}
 We reduce from a special version of the 
\textsf{1-in-3}-satisfiability problem,  called \oitsatversion\ and
introduced in~\cite{FaHeSc}. 
An instance of \oitsatversion\ is a formula of the form 
$\varphi = \bigwedge_{i=1}^n \oit (x^i_1,x^i_2,x^i_3)$, where 
$\oit(x,y,z)$ is a clause which is true iff exactly one 
of the variables $x$, $y$, and $z$ is true. 
Additionally, $\varphi$ has the following properties: 
In each clause the $3$ appearing variables are distinct, and each 
variable appears in exactly $4$ clauses. 
Note that this implies that the number of distinct variables in 
$\varphi$ is $\frac34n$.

An instance $\varphi$ of the problem \oitsatversion\ can be transferred into 
an instance of \textsf{X4C} as follows:
\begin{itemize}
\item Each of the $n$ clauses in $\varphi$ becomes a set element 
      of $\{1,\ldots,n\}$, which we can then rename to values $1,\dots,n$.
\item Each variable $x_i$ becomes a set $S_i$ containing the 
      clauses in which $x_i$ appears.
\end{itemize}

\noindent
First assume that $\varphi$ is satisfiable. 
Then there is an assignment $I$ with $I\models\varphi$. 
Since $I$ satisfies exactly one variable in each clause, we know 
that $n$ variable occurrances are satisfied by $I$. 
Since each variable, in particular each of the satisfied variables, 
appears in $4$ clauses, we get that $\frac14n$ many variables 
are satisfied by $I$. 
We can naturally interpret $I$ as the set of indices $i$ with 
$I\models x_i$ and claim that $I$ satisfies the conditions 
of \textsf{X4C}. 
As mentioned above, since $\card{I}=\frac14n$, it suffices to 
show that $\bigcup_{i\in I} = \{1,\ldots,n\}$. 
This follows from the construction: 
Since $I$ (seen as a truth assignment to the variables) satisfies 
each clause, we know that for each clause, there is a variable 
satisfied by $I$. 
For the \textsf{X4C} instance, this implies tha for each element 
$i\in \{1,\ldots,n\}$, there is an index $j\in I$ with $i\in S_j$.

For the converse, assume that there is an index set $I$ satisfying 
the conditions of \textsf{X4C}. 
We can interpret $I$ as a truth assignment for the variables in 
$\varphi$ in the obvious way: 
A variable is set to $1$ iff its corresponding set is in the 
selection $I$. 
We show that $I$, seen as a truth assignment, satisfies the formula 
$\varphi$. 
Hence let $\oit(x^i_1,x^i_2,x^i_3)$ be a clause in $\varphi$. 
Since $I$ is a set cover, we know that for this clause, an element 
containing the set element corresponding to the clause is selected 
in $I$. 
Hence $I$ satisfies at least one of the variables $x^i_1$, $x^i_2$, 
and $x^i_3$. 
Since $I$ is an exact cover, we also know that each set element 
appears only in one of the selected sets, hence only one of the 
variables is true, and we are done. \hfill $\Box$
\end{proof}

\noindent
We can now show the main theorem of this section.

\begin{Theorem}
For Condorcet voting the constructive control problem by deleting 
voters is NP-hard if the uncovered alternatives are specified as
the winners.
\end{Theorem}

\begin{proof}
We reduce from \textsf{X4C}, which is NP-hard due to Lemma~\ref{lemma:x4c}. 
So, let an \textsf{X4C}-instance consisting of the sets 
$S_1,\dots,S_{\frac34n}$ be given. 
Without loss of generality we assume $n \ge 16$. 
{From} the instance, we construct an election $E$ as follows.
First we define $t = \frac14n-2$ (recall that in every instance 
to \textsf{X4C}, $n$ is a multiple of $4$, hence $t$ is always 
an integer). 
Next we introduce alternatives $a^*$ (the alternative that has to win), 
$s_1,\dots,s_n$ and $b_1,\dots,b_n$.
Finally, we introduce the following four groups of individual
preferences, where
$S_{\neq i} = \set{s_j\ \vert\ j\neq i}$,
$B_{\neq i} =\set{b_j\ \vert\ j\neq i}$,
$B_{\notin S_i} = \set{b_j\ \vert\ j\notin S_i}$ and
$B_{\in S_i} = \set{b_j\ \vert\ j\in S_i}$.

\begin{itemize}
\item[1.] For each $i \in \set{1,\dots,n}$ we use $t$ linear strict orders of the 
      form $S_{\neq i} > s_i > b_i > B_{\neq i} > a^*$.
\item[2.] For each $i \in \set{1,\dots,n}$ we use $t$ linear strict orders of the 
      form $B_{\neq i} > a^* > s_i > b_i > S_{\neq i}$.
\item[3.] For each set $S_i$ we use a linear strict order of the form 
      $B_{\notin S_i} > a^* > S > B_{\in S_i}$.
\item[4.] We use a linear strict order of the form $a^* > S > B$.
\end{itemize}

\noindent
The notation of preferences (linear strict orders) using sets means that 
the order of the alternatives inside the sets is irrelevant. 
For instance, $a^* > S > B$ means that in the linear strict order $a^*$ is the 
greatest element, then the alternatives $s_1,\dots,s_n$ follow in any order 
and, finally, the alternatives $b_1,\dots,b_n$ follow, again in any 
order.
Note that, by definition of \textsf{X4C}, we get $\card{B_{\notin S_i}}=n-4$ 
and $\card{B_{\in S_i}}=4$.
Now, the question in our constructed instance of the control problem is 
whether the specific alternative $a^*$ can be made 
uncovered by deleting at most $\frac14n$ linear strict orders 
(i.e., voters).

We first study the relationship between each of the relevant alternatives
in the constructed election before any deletion of voters is performed.
Note that if the point difference between two alternatives is at 
least $\frac14n+1$, then deleting at most $\frac14n$ linear strict orders
cannot change which of these alternatives dominates the other.
\ifcompleteproof
\else
The following relationships can be verified with a 
case discinction:
\fi

\ifcompleteproof
\begin{description}
\else
\begin{itemize}
\fi

\pointsItem{Each $b_i$ beats $a^*$ with at least $\frac14n+1$ points.} 
\ifcompleteproof
     To see this, we consider all preferences introduced in the election. 
     For each $j\neq i$, the $2t$ linear strict orders of the first two groups
     place $b_i$ ahead of $a^*$. 
     {From} the linear strict orders introduced for $i$, one puts $b_i$ ahead 
     of $a^*$ and the other puts $a^*$ ahead of $b_i$. 
     We now consider the linear strict orders introduced for the sets $S_j$: 
     There are $3$ sets $S_j$ in which $i$ appears (these place 
     $a^*$ ahead of $b_i$) and $i$ does not appear in the remaining 
     $\frac34n-3$ many (these place $b_i$ ahead of $a^*$). 
     Finally, $a^* > S > B$ put $a^*$ ahead of $b_i$.
     Hence the lead of $b_i$ over $a^*$ is  
     $$
     (n-1)\cdot 2\cdot t+\frac34n-6-1=2(n-1)t+\frac34n-7
     $$
     which is at least $\frac14n+1$, since we assumed $n\ge 16$.
\fi

\pointsItem{Alternative $a^*$ beats each $s_i$ with at least $\frac14n+1$ points.}
\ifcompleteproof     
     Note that half of the linear strict orders introduced in the first two 
     groups place $s_i$ ahead of $a^*$ and the other half put 
     $a^*$ ahead of $s_i$. 
     Hence $a^*$ and $s_i$ tie in the sub-election consisting of 
     these linear strict orders. 
     In the $\frac34n$ linear strict orders introduced for the sets $S_i$, 
     however, $a^*$ is always placed ahead of $s_i$. 
     Finally, $a^*$ is ahead of $s_i$ in $a^* > S > B$. 
     As a consequence $a^*$ beats each $s_i$ with $\frac34n+1$ many points, 
     which is at least $\frac14n+1$.
\fi

\pointsItem{If $i\neq j$, then $b_i$ beats $s_j$ with at least $\frac 14n+1$ points.}
\ifcompleteproof
     To see that this is true, note that the linear strict orders introduced in the 
     first two groups are neutral between $b_i$ and $s_j$, as half 
     of them have $b_i$ ahead of $s_j$ and the other half have 
     $s_j$ ahead of $b_i$ (recall that $i\neq j$). 
     Now consider the linear strict orders introduced for the sets $S_i$. 
     There are $3$ such orders which place $s_j$ ahead of $b_i$ (the ones 
     corresponding to sets $S_l$ with $i\in S_l$), and the remaining 
     $\frac34n-3$ many place $b_i$ ahead of $s_j$ (these are the ones 
     corresponding to sets $S_l$ with $i\notin S_l$). 
     In $a^* > S > B$, $s_i$ is voted ahead of $b_i$. 
     Hence $b_i$ beats $s_j$ by $\frac34n-7$ points, which is at 
     least $\frac14n+1$, since again $n \ge 16$.
\fi

\pointsItem{Alternative $b_i$ beats $s_i$ with exactly $\frac14n-3$ points.}
\ifcompleteproof
     This holds due to the following: 
     The linear strict orders introduced in the first two groups for $j\neq i$ are 
     neutral with respect to the relationship between $s_i$ and 
     $b_i$ (half of them put $s_i$ ahead of $b_i$, the other half 
     put $b_i$ ahead of $s_i$). 
     The $2t$ many linear strict orders introduced for $i$ in the first two groups 
     all put $s_i$ ahead of $b_i$.

     Now we consider the linear strict orders introduced for the sets $S_j$. 
     If $i\in S_j$, then $s_i$ is ahead of $b_i$ here, this 
     happens $3$ times. 
     In the remaining $\frac34n-3$ linear strict orders introduced for the 
     sets $S_j$, 
     we have that $i\notin S_j$ and hence in these linear strict orders, $b_i$ is 
     ahead of $s_i$. 
     In $a^* > S > B$ the alternative $s_i$ is voted 
     ahead of $b_i$. 
     Together we have that $b_i$ beats $s_i$ with  
     $$
     -2t-3+\frac34n-3-1=\frac34n-2t-7
     $$
     votes. 
     Since $t=\frac14n-2$, it follows that $\frac34n-2t-7=\frac14n-3$ 
     as required.
\fi
\ifcompleteproof
 \end{description}
\else
 \end{itemize}
\fi

\noindent
In particular, it follows that by deleting at most $\frac14n$ voters, 
the only relevant relationships that can be influenced are those 
between $b_i$ and $s_i$ (we will see that the relationships between 
$b_i$ and $b_j$ or $s_i$ and $s_j$ for $i\neq j$ are not relevant).

We now show that the reduction is correct: 
The instance of \textsf{X4C} is positive iff $a^*$ can 
be made a winner of the election using the Condorcet criterion 
with uncovered set by deleting at most $\frac14n$ linear strict orders.

First, assume that the instance is positive, and let $I$ be a 
corresponding index set. 
We delete the $\frac14n$ linear strict orders corresponding to the 
elements in $I$ and denote the resulting election with $E'$. 
Then $a^*$ indeed is uncovered in $E'$. 
To show this, it suffices to prove that none of the $b_i$ covers 
$a^*$, since $a^*$ wins against all of the $s_i$ (since $a^*$ leads 
against $s_i$ with at least $\frac14n+1$ linear strict orders, this remains 
true also after deleting at most $\frac14n$ linear strict orders). 
Hence, assume that some $b_i$ covers $a^*$ in $E'$. 
It suffices to prove that $s_i$ dominates $b_i$ in $E'$, then, 
since $a^*$ dominates $s_i$ in $E'$, it follows that $b_i$ does 
not cover $a^*$. 
Note that in the original election $E$ the alternative $b_i$ beats $s_i$ 
with $\frac14n-3$ points. 
Deleting the $\frac14n$ linear strict orders corresponding to $I$ has the 
following effect:

\begin{itemize}
\item[a)] For the deleted linear strict orders corresponding to sets $S_j$ with 
      $i \notin S_j$, the alternative $s_i$ gains a point against $b_i$. 
      Since $i$ appears in exactly one of the chosen and 
      $\frac14n$ linear strict orders are deleted, this means that $s_i$ gains 
      $\frac14n-1$ points against $b_i$ from these linear strict orders.
\item[b)] For the single deleted linear strict order corresponding to a set $S_j$ with 
      $i\in S_j$, the alternative $s_i$ loses a point against $b_i$.
\end{itemize}

\noindent
Hence altogether, $s_i$ gains $\frac14n-2$ points against $b_i$
and, thus, now beats $b_i$ with a single point. 
Therefore, as claimed, $b_i$ does not cover $a^*$.

For the converse direction, assume that it is possible to make 
$a^*$ a winner of the election by deleting at most $\frac14n$ 
linear strict orders. 
Again, let 
$E'$ be the election resulting from $E$ by the deletions. 
Since the relatinship between the $b_i$'s and $a^*$ cannot be 
changed by deleting at most $\frac14n$ linear strict orders and $b_i$ wins 
against $a^*$ in the original election $E$, all $b_i$ also win 
against $a^*$ in $E'$.  
Since $a^*$ is a winner in $E'$, it follows that for each $b_i$ 
there must be a alternative dominating $b_i$ who does not 
dominate $a^*$. 
Since for $i \neq j$, we know that $b_i$ wins against $s_j$ in 
the election $E'$, it follows that for all relevant $i$ the alternative $s_i$ 
wins against $b_i$ in $E'$. 
Since $b_i$ wins against $s_i$ with $\frac14n-3$ points, it 
follows that each $s_i$ must gain at least $\frac14n-2$ 
points against $b_i$ by the removal of linear strict orders. 
Hence $n(\frac14n-2) = \frac14n^2-2n$ points need to be gained 
collectively by all $s_i$ against their corresponding $b_i$. 
Obviously, only deleting linear strict orders introduced for the sets $S_j$ 
helps to let $s_i$ gain points against $b_i$. 
Deleting one of these linear strict orders gains $n-8$ points (since it 
hurts for the $4$ values of $i$ with $i \in S_j$, and helps 
the remaining $n-4$ ones). 
Hence, by deleting $\frac14n$ linear strict orders we can gain at most 
$\frac14n\cdot(n-8) = \frac14n^2-2n$ points. 
Since this is the total number of points that need to be 
gained, we know that exactly $\frac14n$ linear strict orders are deleted to 
obtain the election $E'$, and each of these linear strict orders is one 
introduced for a set $S_j$. 
Now assume that there is some $i$ such that two linear strict orders 
corresponding to sets $S_{j_1}$ and $S_{j_2}$ are deleted, 
where $i \in S_{j_1}$ and $i \in S_{j_2}$ and $j_1 \neq j_2$. 
Then $s_i$ gains a point against $b_i$ in at most $\frac14n-2$ 
of the deletions and loses in at least $2$ of them. 
Hence, $s_i$ gains at most $\frac14n-4$ points against $b_i$ and this
implies that $s_i$ loses against $b_i$ in $E'$, which is a 
contradiction. 
Therefore, it follows that each $i$ is contained in at most 
one of the $S_j$ whose corresponding linear strict order is deleted. 
Due to cardinality reasons ($\frac14n$ linear strict orders corresponding 
to sets of $4$ elements each are deleted), it follows that 
each $i$ appears in exactly one set.
As a consequence, we have obtained a set cover as required. \hfill $\Box$
\end{proof}

\section{Conclusion}\label{Sec6} 

In this paper, we have demonstrated that the relation-algebraic approach can be 
used to solve \NP-hard problems from Social Choice Theory. 
In particular, this shows how Computer Algebra tools can be used to obtain
practical algorithms for hard problems without relying on domain knowledge for
optimizations. Our results support the point of view that proving \NP-hardness
is not sufficient in order to conclude that a voting system is ``safe'' from 
attempts to influence the outcome of an election.
In addition to the execution of algorithms, \RelView\ also provides us with 
visualizations of both the input and output of the algorithms and some further 
features that support scientific experiments, like step-wise execution, test 
of properties and generation of random relations. 
All this makes the approach especially appropriate for prototyping and 
experimentation, and as such very instructive scientific research as well as
for university education.

An interesting open question is whether similar problems from the Social Choice
literature, as for example the manipulation problem mentioned in the introduction,
can also be solved with \RelView\ or other Computer Algebra tools.



\end{document}